\documentclass[12pt]{article}
\usepackage[usenames, dvipsnames]{xcolor}
\usepackage{jcapmod}

\usepackage{tocloft}
\usepackage[normalem]{ulem}
\usepackage[]{todonotes}
\usepackage{verbatim}
\usepackage{mathrsfs}
\usepackage{cleveref}
\usepackage[margin=.8in]{geometry}
\usepackage[utf8]{inputenc}

\usepackage{multirow}
\usepackage{booktabs}

\usepackage{caption}
\usepackage{longtable}
\usepackage{afterpage}
\usepackage{tikz}
\usepackage{setspace,caption}
\usepackage{lipsum}

\usetikzlibrary{matrix,arrows,calc}
\usepackage{bbm}
\usepackage{slashed}
\usepackage{graphicx}
\usepackage{subcaption}
\usepackage{psfrag}
\usepackage{tensor}
\usepackage{fouridx}
\usepackage{enumitem}
\usepackage{longtable}
\DeclareUnicodeCharacter{200B}{}

\usepackage{bm}
\usepackage{mdframed}
\usepackage{multirow}
\usepackage{soul}
\usepackage{bbold}
\usepackage{multicol}
\usepackage{tikz-cd}
\usepackage{rotating}

\usetikzlibrary{arrows}

  \usepackage{mdframed}
\tikzset{
commutative diagrams/.cd,
arrow style=tikz,
diagrams={>=latex}}

\usepackage{amsmath,amsthm}
\usepackage{amssymb}
\usepackage{mathtools}

\usepackage{feynmf}
\usepackage{marvosym}

\usepackage{import}

\usepackage[utf8]{inputenc}
\usepackage{amsmath}
\usepackage{tikz}
\usepackage{extpfeil}
\usepackage{amsfonts}
\usepackage{amssymb}
\usepackage{bm}
\usepackage{graphicx}
\usepackage{array}
\usepackage{lipsum}
\usepackage{float}
\usepackage{multirow}
\usepackage{hhline}
\usepackage{tabularx}
\usepackage{graphicx}
\usepackage{afterpage}
\usepackage{longtable}
\usepackage{epstopdf}
\usepackage{ytableau}
\usetikzlibrary{decorations.markings}
\usepackage[all]{xy}

\usepackage[titletoc]{appendix}
\makeatletter
\newtheorem*{rep@theorem}{\rep@title}
\newcommand{\newreptheorem}[2]{%
\newenvironment{rep#1}[1]{%
 \def\rep@title{#2 \ref{##1}}%
 \begin{rep@theorem}}%
 {\end{rep@theorem}}}
\makeatother
\newtheorem{lemma}{Lemma}[section]
\newreptheorem{lemma}{Lemma}

\newtheorem{corollary}[lemma]{Corollary}
\newtheorem{prop}[lemma]{Proposition}
\newtheorem{conj}[lemma]{Conjecture}
\newreptheorem{conj}{Conjecture}
\newtheorem{claim}{Claim}[section]
\newreptheorem{claim}{Claim}

\theoremstyle{definition}

\newtheorem{question}[lemma]{Question}

\theoremstyle{remark}
\newtheorem*{rem*}{Remark}

\newcommand{\jht}[1]{{}}
\newcommand{\jt}[1]{{}}

\definecolor{cobalt}{RGB}{44, 98, 120}
\definecolor{celadon}{rgb}{0.67, 0.88, 0.69}
\definecolor{dm}{cmyk}{.20, 0, .30, 0}
\definecolor{burgundy}{rgb}{0.5, 0.0, 0.13}
\definecolor{plotBlue}{RGB}{94, 130, 181}

\newcommand*\xoverline[2][0.75]{
    \sbox{\myboxA}{$\m@th#2$}
    \setbox\myboxB\null
    \ht\myboxB=\ht\myboxA
    \dp\myboxB=\dp\myboxA
    \wd\myboxB=#1\wd\myboxA
    \sbox\myboxB{$\m@th\overline{\copy\myboxB}$}
    \setlength\mylenA{\the\wd\myboxA}
    \addtolength\mylenA{-\the\wd\myboxB}
    \ifdim\wd\myboxB<\wd\myboxA
       \rlap{\hskip 0.5\mylenA\usebox\myboxB}{\usebox\myboxA}%
    \else
        \hskip -0.5\mylenA\rlap{\usebox\myboxA}{\hskip 0.5\mylenA\usebox\myboxB}%
    \fi}
\makeatother

\pdfoutput=1
\begin{document}

\newcommand{\main}{.}
\begin{titlepage}

\setcounter{page}{1} \baselineskip=15.5pt \thispagestyle{empty}

\bigskip\

\vspace{2cm}
\begin{center}
{\Huge 
On the Frozen F-theory Landscape}
\end{center}

\vspace{1cm}

\begin{center}
David R. Morrison$^{1,2}$ and Benjamin Sung$^1$

\vspace{1 cm}

\emph{$^{1}$Department of Mathematics\\University of California, Santa Barbara, CA 93106, USA}

\vspace{.3cm}
\emph{$^{2}$Department of Physics\\University of California, Santa Barbara, CA 93106, USA}

\end{center}

\vspace{1cm}
\noindent

\begin{abstract}
\ We study $6$d $\mathcal{N} = (1,0)$ supergravity theories arising in the frozen phase of F-theory. For each of the known global models, we construct an F-theory compactification in the unfrozen phase with an identical non-abelian gauge algebra and massless matter content. Two such low energy effective theories are then distinguished through gauge enhancements in moduli space. We study potentially new global models obtained via compact embeddings of a plethora of $6$d $\mathcal{N}= (1,0)$ superconformal field theories and little string theories constructed using frozen $7$-branes. In some cases, these provably do not exist, and in other cases, we explicitly construct a compact embedding, yielding $6$d supergravity theories with new massless spectra. Finally, by using gravitational anomaly cancellation, we conjecture the existence of localized neutral hypermultiplets along frozen $7$-brane loci.

\end{abstract}

\end{titlepage}

\clearpage

\tableofcontents
\newpage

\section{Introduction}
The landscape of F-theory compactifications~\cite{Vafa:1996xn,Morrison:1996na,Morrison:1996pp} on elliptic Calabi-Yau manifolds has given rise to the largest class of supersymmetric string vacua to date. Compactifications yielding $8$d $\mathcal{N} = 1$ vacua correspond to elliptic $K3$-surfaces, which are completely classified. On the other hand, $4$d $\mathcal{N} = 1$ vacua correspond to the choice of an elliptic Calabi-Yau fourfold, together with the specification of fluxes and $D3$-brane data. In stark contrast, our understanding of such vacua is largely bottom-up; there exists only algorithmic and computational approaches towards classifying the bases of such fourfolds, together with recent mathematical results on boundedness~\cite{MR4292177}. Moreover, in such ensembles, we have control only on the coarsest aspects of the physics, such as the gauge symmetry. Finally, $6$d $\mathcal{N} = 1$ vacua occupy a rather rich middle ground; there is a sharp link between the physics and the geometry of an elliptic Calabi-Yau threefold. The classification of surface bases is significantly more tractable, the massless content of the corresponding supergravity theories can be completely determined, and there is a rich, growing interplay between swampland constraints and the known boundaries of the landscape.

Nevertheless, there are known string compactifications in other duality frames which involve intricate non-geometric data, even in higher dimensions. For example, in an M-theory background on a $\mathbb{C}^2/\Gamma$ singularity, an additional $C_3$-field period can be specified on the boundary, leading to frozen singularities~\cite{deBoer:2001wca} with a restricted deformation space. The classification of compact backgrounds with such singularities and their low energy effective physics have largely been ignored, due in part to a lack of understanding of frozen singularities, as well as the fact that such vacua are rather special and are somewhat insignificant when searching for generic phenomena in the landscape.

A resurgence in their study was initiated in~\cite{Tachikawa:2015wka}, which demonstrated that the only new supersymmetric defect in F-theory from frozen singularities is a strongly coupled version of the orientifold plane, $O7^+$, with positive $D7$-brane charge. The study of $8$d F-theory vacua with $O7^+$-planes first appeared in~\cite{Witten:1997bs}, and a systematic classification of such compact models was investigated more recently in~\cite{Cvetic:2022uuu}. The effects of incorporating such defects in $8$ dimensions are markedly simple; they lead to either a rank $4$ or a rank $12$ vector multiplet moduli space, reduced from the usual $20$ dimensional space of vacua. In addition, they give rise to sympletic ($\mathfrak{sp}(n)$) gauge algebras, which does not occur in conventional $8$d F-theory compactifications.

The study of F-theory compactifications to $6$-dimensions with $O7^+$-planes was initiated in~\cite{Bhardwaj:2018jgp}, and a number of consistent compact and non-compact constructions~\cite{Bhardwaj:2019hhd} were given based on dualities to other frames. On the other hand, the complete list of conditions that are necessary and sufficient to specify a frozen $6d$ F-theory vacua and the resulting supergravity spectrum is still unknown. In this paper, we take a first step towards understanding the massless spectra of frozen F-theory compactifications to six dimensions. Our aim is to identify a proxy to distinguish the frozen phase from the unfrozen phase of $6d$ F-theory purely from the low energy spectrum, to construct compact embeddings of all known local frozen models, and to compare the resulting supergravity spectrum with the conventional $6d$ F-theory landscape. In the process, we will demonstrate the existence of $6d$ {\it supergravities} with new massless spectra arising from the frozen phase.

We emphasize that this is merely a first step towards understanding $6d$ F-theory compactifications with frozen $7$-branes. The ultimate goal in this direction is the compilation of precise necessary and sufficient conditions for specifying such compactifications, which would lead to a systematic construction of frozen vacua. It would also be desirable to investigate phase transitions among such models, and to make precise the relation with other interpretations of frozen vacua as appearing in~\cite{Bershadsky:1998vn} through the specification of B-field fluxes.

\subsection{Summary of Results}
In this paper, we study the massless supergravity spectrum in the frozen phase of F-theory. We emphasize that there are no known necessary and sufficient conditions guaranteeing the existence of a given F-theory model with frozen $7$-branes. Nevertheless, we will study a plethora of both local and global constructions, whose existence is strongly evidenced either through dualities or through a weakly coupled type IIB limit.

As explained above, a critical new feature in six dimensions is that the frozen and unfrozen phases of F-theory are, a priori, virtually indistinguishable at the level of the low energy effective physics. This leads us to the central question of study in this paper:

\begin{question}\label{q:distinguish}
To what extent can we distinguish between the frozen and unfrozen phases of F-theory from the non-abelian gauge sector and massless matter in the effective $6d$ $\mathcal{N} = (1,0)$ supergravity theories?
\end{question}

\begin{enumerate}
\item
Section~\ref{sec:ordinaryf}: For each known frozen $6$d F-theory model, we explicitly construct an unfrozen $6$d F-theory model with an identical massless spectrum. We give a criterion for distinguishing a frozen and an unfrozen supergravity theory with identical massless spectra by gauge enhancements in moduli space.
\item
Section~\ref{sec:frozenscft}: We prove that most of the known new constructions of superconformal theories using frozen $7$-branes do not admit a compact embedding. In the remaining special cases, we explicitly construct a compact embedding. As a consequence, we construct compact 6d F-theory vacua with new massless spectra in the frozen phase. 
\item
Section~\ref{sec:neutrals}: We conjecture the existence of localized neutral hypermultiplets along intersections of residual $I_1$-loci with frozen $7$-branes.
\end{enumerate}
Our bottom-up analysis of the known constructions in the frozen phase of F-theory lead us to the following statements, which comprise the main results of this paper.

We first study the known compact frozen F-theory models, whether they can be realized in the unfrozen phase, and what criteria are sufficient to distinguish between the two supergravity theories.
\begin{repclaim}{clm:limit}
Let $\mathcal{T}_1, \mathcal{T}_2$ be two $6d$ $\, \mathcal{N} = (1,0)$ supergravity theories with the same massless spectrum, engineered via a frozen/unfrozen F-theory compactification respectively. Assume that there exists an $\mathfrak{sp}(n)$ factor in both theories, which is localized on a frozen $7$-brane in the case of $\mathcal{T}_1$. 

Then $\mathcal{T}_2$ admits a limit in moduli space enhancing $\mathfrak{sp}(n)$ to $\mathfrak{so}(2n)$, while the analogous limit for $\mathcal{T}_1$ realizes a $(4,6)$-divisor.
\end{repclaim}
The existence of a $(4,6)$-divisor corresponds to an infinite distance limit, often realizing a decompactification of the effective $6$d theory. We note that such limits have been analyzed and sharpened in the frozen phase of F-theory in $8$ dimensions in~\cite{Cvetic:2022uuu}.

A natural question is whether the massless spectrum of a $6$d supergravity theory realized in the frozen phase of F-theory, can be realized via an F-theory compactification in the unfrozen phase as well. Given the wealth of constructions in~\cite{Bhardwaj:2019hhd}, we answer this question.

\begin{repclaim}{clm:cases}
Most, but not all of the new superconformal field theories and little string theories exhibited in~\cite{Bhardwaj:2019hhd} do not admit a compact embedding.

Nevertheless there exists compact frozen F-theory compactifications whose massless spectrum cannot be realized in the unfrozen phase of F-theory.
\end{repclaim}
Finally, based on our study of compact examples together with gravitational anomaly cancellation, we make the following conjecture.
\begin{repconj}{conj:neutrals}
Consider a $6$d F-theory compactification with a frozen $7$-brane localized along a divisor $D$ with residual discriminant $\widetilde{\Delta}$. Then there exists at least $\frac{1}{2}D \cdot \widetilde{\Delta}$ localized neutral hypermultiplets.
\end{repconj}

We conclude that the presence of $O7^+$-planes in $6$d F-theory compactifications lead to a number of implications for the low energy effective physics. Most notably, as demonstrated in Claim~\ref{clm:limit}, these lead to a certain obstruction for gauge enhancement in comparison with the unfrozen phase of F-theory, and this will be explored further in future work. In addition, we conjecture that the presence of such orientifold planes could potentially be a new source of localized neutral hypermultiplets, and it would be interesting to investigate this by establishing a weakly coupled type IIB limit in the frozen phase.

%In addition, such defects lead to additional $7$-brane loci
%We give evidence for Conjecture~\ref{conj:samespectrum} with the results of Sections~\ref{sec:ordinaryf} and~\ref{sec:frozenscft}.
%In section~\ref{sec:consistency}, we review the rules of~\cite{Bhardwaj:2018jgp}

%In section~\ref{sec:ordinaryf}, we revisit some models in~\cite{Bhardwaj:2018jgp} from the perspective of ordinary F-theory.

%In section~\ref{sec:frozenscft}, we demonstrate that the new scfts realized in the frozen phase in \cite{Bhardwaj:2019hhd} do not admit a compact embedding.
\section{Consistency conditions for frozen F-theory models}\label{sec:consistency}
In this section, we review the necessary conditions in specifying a compact, frozen $6d$ F-theory model, first appearing in~\cite{Bhardwaj:2018jgp}. In Section~\ref{sec:Ftheory}, we review the basics of $6d$ F-theory compactifications, emphasizing a perspective which generalizes well in defining compactifications with frozen $7$-branes. In Sections~\ref{sec:frozen} and~\ref{sec:frozenconstruct}, we review the basics of frozen singularities, their realization in F-theory models, and present an example.
\subsection{6d F-theory compactifications}\label{sec:Ftheory}
In this section, we review the basics of F-theory compactifications to six dimensions.

An F-theory compactification to six dimensions is specified by the choice of an elliptically fibered Calabi-Yau threefold $\pi \colon X \rightarrow B$ over a smooth complex algebraic surface $B$. Up to a sequence of birational transformations~\cite{NAKAYAMA1988405}, this data can be specified by a Weierstrass equation
\begin{equation}\label{eq:weierstrass}
y^2 = x^3 + fx + g, \qquad \Delta \colon 4f^3 + 27g^2 = 0
\end{equation}
where $f$ and $g$ are global sections of $-4K_B$ and $-6K_B$ respectively, with $-K_B$ denoting the anti-canonical divisor on $B$.
 The torus fiber degenerates along the algebraic curve $\Delta \subset B$, often called the discriminant locus, and in generic situations, the discriminant decomposes into a union of smooth irreducible algebraic curves $\Delta = \Delta_1 \cup \ldots \cup \Delta_n$. In particular, loops around the codimension-$1$ irreducible components $\Delta_i$ correspond to elements in $SL(2,\mathbb{Z})$, which naturally induce a monodromy action on the $1$-cycles of the torus fiber.

The low energy effective physics is specified by a $6$d $\mathcal{N} = (1,0)$ supergravity theory, and such theories occupy a particularly rich corner of the string landscape. On one hand, they are strongly constrained by gauge, gravitational, and mixed anomalies, and on the other hand, even the possible combinations of non-abelian gauge algebras and massless matter content has evaded a complete classification. Our knowledge of the possible supergravity theories resulting from string compactifications has progressed in tandem with our understanding of elliptic Calabi-Yau threefolds, whose singular limits have not been completely classified. We will not attempt to summarize the various constraints on $6d$ $\mathcal{N} = (1,0)$ supergravity theories and we point to the vast literature~\cite{Kumar:2009ac,Kumar:2010ru,Morrison:2011mb,Morrison:2012np,Morrison:2012js,Kumar:2010am,Kumar:2009ae,Kumar:2009us,Park:2011wv,Monnier:2017oqd} for a comprehensive treatment.

The details of the effective $6$d $\mathcal{N} = (1,0)$ supergravity theory rely crucially on the singularities of the total space $X$. Indeed, codimension-$1$ irreducible components $\Delta_i$, specify the location of stacks of generalized $(p,q)$ $7$-branes, which can support non-abelian gauge algebras. At codimension-$2$, pairwise intersections of the discriminant components $\Delta_i \cap \Delta_j$ support localized modes corresponding to jointly charged matter. The precise non-abelian gauge group can be deduced from the standard classification of Kodaira-singular fibers together with an analysis of their global monodromies, and we refer to~\cite[Tables 1, 2]{Grassi:2014zxa} for details. The precise counting of jointly charged hypermultiplets is somewhat more subtle and depends on details of the intersections, and we refer to~\cite{Grassi:2011hq} for details.

In general, the gauge group and matter content of a $6$d F-theory compactification is often studied via M-theory compactified on a crepant resolution $\widehat{X} \rightarrow X$, and then inferred by taking the singular limit. Motivated by the non-existence of crepant resolutions in certain examples as well as an intrinsic study of the gauge theory at the singular point, a theme~\cite{Grassi:2018wfy, Grassi:2021ptc} has emerged which aims to develop techniques to study F-theory directly on $X$. In particular, the $6$-dimensional physics is sensitive only to the associated Weierstrass model~(\ref{eq:weierstrass}), and not to the details of the total space $X$. This is exemplified through the fact that the $SL(2,\mathbb{Z})$-monodromy associated to an irreducible component of the discriminant $\Delta_i$ admits a (unique up to conjugation) factorization~\cite[Section 3.2]{Grassi:2018wfy} and~\cite{velez2008normal}:
\[
M_i = M_{p_1, q_1} \cdot \ldots \cdot M_{p_n,q_n}, \quad M_{p,q} = \begin{pmatrix}
1 - p q & p^2 \\
-q^2 & 1+ p q
\end{pmatrix}
\]
Put differently, every discriminant component $\Delta_i$ can be decomposed uniquely into a combination of $(p,q)$ $7$-branes and we will use this perspective in our analysis of the frozen phase of F-theory.
\subsection{Frozen $7$-branes}\label{sec:frozen}
A somewhat less understood class of singularities in M-theory arise from singularities of the form $\mathbb{C}^2 / \Gamma$, together with the specification of a fractional $C_3$-field period on a boundary $3$-sphere. These so-called {\it frozen singularities}~\cite{deBoer:2001wca} have a restricted deformation space and exhibit a smaller gauge algebra than expected from the singularity itself. 

Though understood for quite some time~\cite{Berkooz:1996iz}, $O7^+$-planes (orientifold planes with positive RR-charge), were demonstrated in~\cite{Tachikawa:2015wka} to be the only class of frozen singularities arising in F-theory. This suggests that $O7^+$-planes are the only remaining class of supersymmetric defects in F-theory compactifications, and motivates a more systematic study of their effects on compactifications to lower dimensions. We briefly review their critical properties, summarizing the similarities and differences with the usual combinations of $(p,q)$ $7$-branes, and we refer to~\cite{Bhardwaj:2018jgp} for details. 

A single $O7^+$-plane shares a number of similarities with a configuration consisting of a single $O7^-$-plane parallel to $8$ $D7$-branes. As explored at length in~\cite{Witten:1997bs}, both exhibit $+4$ units of $D7$-brane charge, and a $D3$-brane encircling such defects exhibit a monodromy $SL(2,\mathbb{Z})$-conjugate to that of an $I_4^*$-fiber:
\[
M_{I_4^*} = \begin{pmatrix}
-1 & 4 \\
0 & 1
\end{pmatrix}
\]
Nevertheless, there are a number of critical differences; for one, an $O7^+$-plane on top of $n$ $D7$-brane yields an $\mathfrak{sp}(n)$ gauge symmetry, while an $O7^-$-plane on top of $n+4$ $D7$-branes yields an $\mathfrak{so}(2n+8)$ gauge symmetry. As a consequence, a single $D3$-brane on top of a single $O7^+$-plane may fractionate into two, which does not happen in the latter case, due to a factor of two in the relative Chan-Paton indices. Moreover, a $D3$-brane colliding with an $O7^+$-plane yields an $O(2)$-gauge theory with two charge $2$ massless hypermultiplets, while a collision with the latter configuration yields $SU(2)$ $N_f= 8$ Seiberg-Witten theory. These observations are in line with the results of~\cite{Cvetic:2022uuu}, which concluded that $(p,q)$-strings must end on an $O7^+$-plane with both even $p$ and $q$ charges in contrast to the latter case, where $p$ and $q$ can take arbitrary integer values.

\subsubsection{F-theory constructions}\label{sec:frozenconstruct}
As emphasized in Section~\ref{sec:Ftheory}, we will specify an F-theory construction through the choice of a smooth algebraic surface $B$, together with the location of $(p,q)$ $7$-branes localized along curves in $B$. Similarly, an F-theory construction in the frozen phase will be specified through the same information, together with the choice of replacing a stack of $(p,q)$ $7$-branes with the monodromy of an $I_{n+4}^*$-fiber by an $O7^+$-plane with $n$ $D7$-branes. Following~\cite{Bhardwaj:2018jgp}, we will denote the latter singularity with $\widehat{I}_{n+4}^*$.

On the other hand, as argued via duality in~\cite{Bhardwaj:2018jgp}, such an F-theory construction does not admit a canonical assignment of divisors supporting individual gauge algebras. In this section, we will review the necessary conditions in specifying a frozen $6$d F-theory compactification, but we emphasize that these are not sufficient. Throughout this section, we will fix a smooth compact algebraic surface $B$ with canonical divisor $K$, and Weierstrass polynomials $f$ and $g$. The associated discriminant $\Delta$ contains irreducible components $\Delta_a$, with corresponding monodromies $M_a$. 

A $6$d F-theory compactification in the frozen phase is specified by the following data:
\begin{enumerate}[label=(\alph*)]
\item
For each irreducible component $\Delta_a$ with monodromy matrix $M_a$ conjugate to $M_{I_n^*}$, the choice of an unfrozen/frozen $7$-brane, designated by $I_n^*$/$\widehat{I}_n^*$ respectively. We denote by $F$, the sum of irreducible divisors $D_a$ supporting a frozen $7$-brane.
\item
A collection of gauge algebras $\mathfrak{g}_i$, together with disjoint embeddings $\rho_{i,a} \colon \mathfrak{g}_i \xhookrightarrow{} \mathfrak{l}_a$ such that $\bigoplus_i \rho_{i,a}(\mathfrak{g}_i) \subset \mathfrak{l}_a$
\end{enumerate}
We define the associated {\it gauge divisor} by
\[
\Sigma_i \coloneqq \sum\limits_a \mu_{i,a} o_{i,a} D_a
\]
where $\mu_{i,a} = 0$ if $\rho_{i,a} = 0$, and is $1$ if $\rho_{i,a}$ is non-trivial. The coefficient $o_{i,a}$ is the Dynkin index of the embedding $\rho_{i,a}$ and we refer to~\cite{MR47629} for a precise definition. 

By invoking the Green-Schwarz mechanism in the presence of $O7^+$-planes, one finds the following modified anomaly polynomial
\[
I_{GS}^8 = -\frac{1}{2} \big(-(K+F) \frac{p_1(T)}{4} + \sum\limits_i \Sigma_i \nu(F_i)\big)^2
\]
where $p_1(T)$ is the pontryagin class of the tangent bundle of $B$ and $\nu(F_i)$ is the instanton number density of the field strength $F_i$ valued in $\mathfrak{g}_i$. This agrees with the usual anomaly polynomial when $F = 0$ and strongly suggests replacing the usual assignment of matter content with the simple substitution
\[
K \mapsto K +F 
\]
which leads to the following additional consistency conditions. We will demand that the following must hold to define a consistent frozen F-theory configuration:
\begin{enumerate}
\item
$\Sigma_i \cdot \Sigma_i \in \mathbb{Z}$, $\Sigma_i \cdot \Sigma_j \in \mathbb{Z}_{\geq 0 }$ \text{ for all } $i,j$
\item
$(K+F) \cdot \Sigma_i + \Sigma_i \cdot \Sigma_i = -2$
\item
For each pair $(\mathfrak{g}_i,\Sigma_i)$, the total number of all localized charged hypermultiplets must add up to the total number of hypermultiplets of corresponding representation, dictated by~\cite[Table 3.1]{Bhardwaj:2018jgp}.
\item
$n_{H,ch} - n_V < 273$
\end{enumerate}
Roughly, the first three conditions correspond to the cancellation of gauge anomalies, while the last imposes gravitational anomaly cancellation. We will assume that all intersections of gauge divisors are transversal, each of which hosts a single bi-fundamental, aside from $\mathfrak{so}-\mathfrak{sp}$ intersections which host a half bi-fundamental.

We briefly summarize the Aspinwall-Gross model, which first appeared in~\cite{Aspinwall:1996nk}, and admits a frozen variant described in~\cite{Bhardwaj:2018jgp}. Consider the base $B = \mathbb{F}_4$, with Weierstrass polynomials
\begin{align*}
f = e^2 \widetilde{f}, \quad g = e^3 \widetilde{g},  \quad \Delta = e^{18} q^{48} \widetilde{\Delta}
\end{align*}
where $\widetilde{f}, \widetilde{g}, \widetilde{\Delta}$ are generic members of the prescribed class. Moreover, $e$ denotes the exceptional divisor, and $q$ denotes the fiber class of $B$. We note that a concrete tuning realizing such polynomials on $B$ can be found in~\cite{Aspinwall:1996nk}.

In the conventional F-theory phase, $e$ supports an $I_{12}^*$-singularity with an $\mathfrak{so}(32)$-algebra, $q$ supports an $I_{48}^{ns}$-singularity with an $\mathfrak{sp}(24)$-algebra, and there is a single bi-fundamental localized at the intersection. Following the above prescription, we now assume that $e$ supports a frozen $\widehat{I}_{12}^*$-singularity. In addition, we define the following frozen and gauge divisors
\[
F = e, \quad \Sigma_1 = \frac{1}{2}e, \quad \Sigma_2 = 2q 
\]
where $\Sigma_1$ supports an $\mathfrak{sp}(8)$ gauge algebra, and $\Sigma_2$ supports a $\mathfrak{su}(24)$ gauge algebra. We note that the canonical embedding $\mathfrak{su}(24) \xhookrightarrow{} \mathfrak{su}(48)$ is an embedding of index $2$, contributing a factor of $2$ in definition of $\Sigma_2$. Finally we note that anomaly cancellation implies the following
\[
(K+F) \cdot \Sigma_1 = -1, \quad \Sigma_1 \cdot \Sigma_1 = -1, \quad (K+F) \cdot \Sigma_2 = -2, \quad \Sigma_2 \cdot \Sigma_2 = 0, \quad \Sigma_1 \cdot \Sigma_2 = 1
\]
which gives $1$ bifundamental of $\mathfrak{sp}(8) \times \mathfrak{su}(24)$ and $2$ antisymmetrics of $\mathfrak{su}(24)$. We note that the gauge group appearing on the fiber class $q$ in this compact case, is different from the expected gauge group of transverse $D7$-branes intersecting an $O7^+$-plane. 
\section{Frozen models via ordinary F-theory embeddings}\label{sec:ordinaryf}
In this section, we study in detail the massless spectrum of the compact models constructed in the frozen phase of F-theory in~\cite{Bhardwaj:2018jgp}. These models admit alternative constructions via duality to either the heterotic or type I string, and thus are particularly credible examples with which to explore consistency conditions for constructing frozen F-theory compactifications. In order to begin identifying universal features appearing in the frozen phase, for every such example, we will explicitly produce an unfrozen F-theory model which engineers a $6$d $\mathcal{N} = (1,0)$ supergravity theory with precisely the same massless content. For simplicity, we will ignore potential contributions from $\mathfrak{u}(1)$ factors, as well as the global forms of the gauge group.

In Section~\ref{sec:aspinwallgross}, we will construct an F-theory compactification on $\mathbb{F}_1$ reproducing the Aspinwall-Gross model with a frozen $7$-brane on $\mathbb{F}_4$. In Section~\ref{sec:1frozen}, we will perform a similar construction for the frozen compact $\mathbb{P}^1 \times \mathbb{P}^1$ examples. Such constructions are necessarily more involved, and will proceed in two steps:
\begin{enumerate}
\item
We first tune the required gauge algebras on curves with the correct mutual intersections on $\mathbb{P}^2$ using a Tate model.
\item
We then demonstrate that there exists blowups to a compact base with the correct number of tensor multiplets and matter content.
\end{enumerate}
In Section~\ref{sec:distinguishing}, we identify a criterion to distinguish between unfrozen and frozen F-theory compactifications realizing the same massless spectrum.

\subsection{Tate forms and $(4,6)$-points}\label{sec:tate}
In this section, we briefly review a more generalized form of the Weierstrass equation~(\ref{eq:weierstrass}) in engineering elliptic Calabi-Yau threefolds. The main utility of such constructions is that it allows us to engineer rather degenerate singular configurations using the results of~\cite{Katz:2011qp}. 
%$SU(4)$ on genus $1$ curve in blowup of $\mathbb{P}^2$ at $5$ points

In general, one can construct an elliptic Calabi-Yau threefold using a Tate model:
\begin{equation}\label{eq:tate}
y^2 + a_1xy+ a_3 y = x^3 + a_2 x^2 + a_4 x + a_6
\end{equation}
where the above coefficients are global sections of appropriate powers of the anti-canonical class on $B$, in particular, $a_n \in H^0(X, \mathcal{O}(- n K_B))$. Using \cite[Tables 2, 3]{Katz:2011qp}, we can easily engineer Kodaira $I_n$ and $I_m^*$-fiber types which are critical in constructing frozen singularities.

In order to interpolate between the Tate and Weierstrass form, we will rely on the results described in \cite{Lawrie:2012gg}. The Weierstrass polynomials corresponding to the Tate form~(\ref{eq:tate}) are given by the following
\begin{equation}\label{eq:tateweierstrass}
\begin{split}
f &= \frac{1}{48}(-(a_1^2 + 4a_2)^2 + 24 a_1 a_3 + 48 a_4) \\
g &= \frac{1}{864}(-(a_1^2 + 4a_2)^3 + 36 (a_1 a_3 + 2a_4)(a_1^2 + 4a_2) -216 (a_3^2 + 4a_6)) 
\end{split}
\end{equation}
The associated discriminant $\Delta$ can then be readily computed using~(\ref{eq:weierstrass}).

In practice, we will often engineer the desired singularities using a Tate form, and then achieve the desired model through a sequence of phase transitions by blowing up the base of the elliptic fibration. In order to preserve the Calabi-Yau condition, it is well understood that one needs to first engineer a point $p \in B$ such that the orders of vanishing of the Weierstrass polynomials satisfy $(ord_p(f), ord_p(g), ord_p(\Delta)) \geq (4,6,12)$. We may then blow up $B$ at $p$ and together with a change of variables, this yields a Calabi-Yau threefold over the blowup. Physically, such orders of vanishing indicate a conformal matter point, and blowing up corresponds to moving out on the tensor branch of the theory. For details, we refer to~\cite[Appendix A]{Apruzzi:2018oge}.

\subsection{Conventional embeddings of Frozen F-theory models}\label{sec:conventional}
\subsubsection{Aspinwall-Gross model}\label{sec:aspinwallgross}
%Realize frozen Aspinwall-Gross model via ordinary F-theory embedding on $\mathbb{F}_1$. How to get $\mathfrak{su}$? 
%\begin{itemize}
%\item
%Construct frozen Aspinwall-Gross using $\mathbb{F}_1$ obstructor instantons
%\item
%Comment on distinct enhancement limits in moduli space (no $\mathfrak{su}$ enhancement of fibers in ordinary F-theory on $\mathbb{F}_1$, no $\mathfrak{so}$ enhancement in frozen F-theory)
%\item
%Argue that $\mathfrak{su}$ imposed by anomaly constraints?
%\end{itemize}
%$I_{12}^*$ on $e$, $I_{48}^{ns}$ on $f$ 
We will take the base surface of the F-theory compactification to be the Hirzebruch surface $\mathbb{F}_4$ with a unique curve of self-intersection $-4$.
In~\cite[Section 5.1]{Bhardwaj:2018jgp}, a frozen variant of the F-theory compactification~\cite{Aspinwall:1996vc} on $\mathbb{F}_4$ with a tuned $I_{12}^*$ fiber along the zero section was given.

The massless spectrum consists of the following
\begin{table}[h!]
\centering
\begin{tabular}{c|ccccccccc}
\toprule
gauge algebra &  $\mathfrak{sp}(8) \oplus \bigoplus\limits_{i=1}^{12} \mathfrak{sp}(1)_i$\\
\midrule
charged hypermultiplets & $\mathbf{16} \otimes \mathbf{2}_i$, $1 \leq i \leq 12$\\
\midrule
uncharged hypermultiplets & 32\\
\midrule
tensor multiplets & 1\\
\bottomrule
\end{tabular}
\caption{Massless matter content of the generic Aspinwall-Gross model in the frozen phase. Throughout, we will omit the supergravity multiplet.}
\end{table}

We will reproduce the same massless content with a conventional F-theory compactification with base surface $B = \mathbb{F}_1$, i.e. the Hirzebruch surface with a unique curve of self-intersection $-1$.

We first consider the following tuning with a Tate model
\begin{equation}\label{eq:agunfrozen}
a_1 = 0, \ a_2 \sim p_2, \ a_3  = 0, \ a_4 \sim e^8 p_4, \ a_6 = 0
\end{equation}
where the $p_i$'s denote generic representatives of the corresponding global section. This particular tuning implies that we must have $p_4 \in \vert 12 f \vert$. This leads to the following Weierstrass model
\begin{align*}
f &= \frac{1}{3} (-p_2^2 + 3 e^8 p_4)\\
g&= -\frac{1}{27} p_2(2p_2^2 - 9e^8 p_4) \\
\Delta &= -e^{16} p_4^2 (p_2^2 - 4e^8 p_4)
\end{align*}
From the form of the discriminant $\Delta$, one verifies straightforwardly that this indeed reproduces the desired massless spectrum. 
\subsubsection{Compact $\mathbb{P}^1 \times \mathbb{P}^1$ with one frozen}\label{sec:1frozen}
In this section, we construct an unfrozen F-theory model reproducing the same massless content as that of the $\mathbb{P}^1 \times \mathbb{P}^1$ model with a single frozen $7$-brane exhibited in~\cite[Section 5.3]{Bhardwaj:2018jgp}. At a perturbative level, these coincide with the models of~\cite{Bianchi:1990tb,Gimon:1996rq}, obtained by compactifying type IIB on $T^2/\mathbb{Z}^2 \times T^2 / \mathbb{Z}^2$, and flipping a single $O7^-$ to an $O7^+$-plane. The model we consider is obtained from the Higgs branch of an F-theory realization of such a compactification.

The massless spectrum for this model is given in Table~\ref{tab:onefrozen}

\begin{table}[h!]
\centering
\begin{tabular}{c|ccccccccc}
\toprule
gauge algebra &  $\mathfrak{su}(8) \oplus \mathfrak{sp}(4)_1 \oplus \mathfrak{sp}(4)_2$\\
\midrule
charged hypermultiplets & $\mathbf{8} \otimes \mathbf{8}_1, \mathbf{8} \otimes \mathbf{8}_2, \mathbf{8}_1 \otimes \mathbf{8}_2 $\\
 & $2 \times \Lambda^2 \mathbf{8}$\\
\midrule
uncharged hypermultiplets & 16\\
\midrule
tensor multiplets & 5\\
\bottomrule
\end{tabular}
\caption{Massless content of a frozen F-theory model on $\mathbb{P}^1 \times \mathbb{P}^1$}
\label{tab:onefrozen}
\end{table}

In a conventional F-theory compactification, anomaly cancellation then implies that $\mathfrak{su}(8)$ must be supported on a $(0)$-curve and $\mathfrak{sp}(4)_i$ must be supported on $(-1)$ curves, with all three intersecting pairwise. We will reproduce this massless content with a base surface $B = \mathbb{P}^2$.

We first consider the following tuning with a Tate model on $\mathbb{P}^2$
\[
a_1 \sim p_1, \ a_2 \sim q^2 h_1^2, \ a_3 = 0, \ a_4 \sim h_1^4h_2^4 h_3^4, \ a_6 = 0
\]
where $h_1, h_2, h_3 \in \vert h \vert $ are general lines. From the degrees of the Tate coefficients, $p_1$ must be a degree $3$ polynomial, and $q$ is a degree $2$ polynomial. We may choose polynomials $p_1$ and $q$ such that they vanish simultaneously at a single point $u$ on the line $h_1$, two points $v_1, v_2$ on $h_2$, and two points $w_1, w_2$ on $h_3$. Moreover, we may choose such points away from the pairwise intersections $h_i \cap h_j$.

Using the ansatz given in equation~(\ref{eq:tateweierstrass}), this leads to the following Weierstrass model
\begin{align*}
f &= \frac{1}{48} (48 h_1^4 h_2^4 h_3^4 - p_1^4 - 8 h_1^2 p_1^2 q^2 - 16 h_1^4 q^4) \\
g &= \frac{1}{864} (p_1^2 + 4 h_1^2 q^2) (72 h_1^4 h_2^4 h_3^4 - p_1^4 - 8h_1^2 p_1^2 q^2 - 16 h_1^4 q^4) \\
\Delta &= \frac{1}{16} h_1^8 h_2^8 h_3^8 (8 h_1^2 h_2^2 h_3^2 - p_1^2 -  4 h_1^2 q^2) (8 h_1^2 h_2^2 h_3^2 + p_1^2 + 4 h_1^2 q^2)
\end{align*}
In particular, one verifies straightforwardly that this realizes an $\mathfrak{su}(8)$ algebra supported on $h_1$ and $\mathfrak{sp}(4)$ algebras supported on $h_2, h_3$, with all three intersecting pairwise. 

Next, we perform the required blowups to arrive at a model with precisely $5$ tensor multiplets. From the choices of $p_1$ and $q$, as well as the above Weierstrass model, we note that the five points $u, v_1, v_2, w_1, w_2$ are all $(4,6)$-points. Blowing up all five points, we conclude that the exceptional divisors do not support any additional gauge algebras, and that this gives the desired self-intersections required by anomaly cancellation.
%$\mathfrak{su}(8)$ on $(2,1)$, $\mathfrak{sp(4)}$ on $(0,1)$ in $\mathbb{P}^1 \times \mathbb{P}^1$: In discriminant $(16,8)$, $(0,8)$, $(8,8)$
\subsubsection{Compact $\mathbb{P}^1 \times \mathbb{P}^1$ with two frozen}\label{sec:2frozen}
Finally, in this section, we construct an unfrozen F-theory model reproducing the same massless content as that of the $\mathbb{P}^1 \times \mathbb{P}^1$ model with two frozen $7$-branes.

The massless spectrum now consists of the following
\begin{table}[h!]
\centering
\begin{tabular}{c | c}
\toprule
gauge algebra &  $\mathfrak{so}(8) \oplus \bigoplus\limits_{i=1}^{4} \mathfrak{sp}(2)_i $\\
\midrule
charged hypermultiplets & $\mathbf{4}_i \otimes \mathbf{4}_j, i < j$\\
\midrule
uncharged hypermultiplets & 13\\
\midrule
tensor multiplets & 8\\
\bottomrule
\end{tabular}
\caption{Massless content of a frozen F-theory model on $\mathbb{P}^1 \times \mathbb{P}^1$}
\label{tab:twofrozen}
\end{table}

We note that anomaly cancellation requires that the $\mathfrak{so}(8)$ factor is supported on a $(-4)$-curve, while the $\mathfrak{sp}(2)$ factors are all supported on $(-1)$-curves.

We consider the following tuning with a Tate model, again on $\mathbb{P}^2$
\[
a_1= 0 , \ a_2 \sim q^2 p_2, \ a_3 =  0 , \ a_4 \sim h_1^2 h_2^2 h_3^2 h_4^2 q^2, \ a_6 = 0
\]
where again, $h_1, h_2, h_3, h_4 \in \vert h \vert$ are four general lines on $\mathbb{P}^2$. Moreover, again from the Tate coefficients, $q$ and $p_2$ are general quadrics. Converting to a Weierstrass model, we obtain the following
\begin{align*}
f &= \frac{1}{3}q^2 (3h_1^2 h_2^2 h_3^2 h_4^2 - p_2^2 q^2) \\
g &= \frac{1}{27}p_2 q^4 (9h_1^2 h_2^2 h_3^2 h_4^2 - 2 p_2^2 q^2 ) \\
\Delta &= -h_1^4 h_2^4 h_3^4 h_4^4 q^6 (-2 h_1 h_2 h_3 h_4 + p_2 q)(2 h_1 h_2 h_3 h_4 + p_2 q)
\end{align*}
it follows straightforwardly that this realizes an $\mathfrak{sp}(2)$ factor on each line $h_i$ all intersecting pairwise, as well as an $\mathfrak{so}(8)$ factor supported on the quadric $q$.

As in Section~\ref{sec:1frozen}, we perform the required blowups to arrive at a model with $8$ tensor multiplets. We note from the form of the Weierstrass model, that the $8$ points of intersections $q \cap h_i$ for $i = 1,2,3,4$ are all $(4,6)$-points. Blowing up all these points, we again conclude that the exceptional divisors do not support any additional gauge algebras, and that this gives the desired self-intersections required by anomaly cancellation.

\subsection{Distinguishing frozen from unfrozen supergravity theories}\label{sec:distinguishing}
Our goal in Section~\ref{sec:conventional} was to construct, for every compact model of~\cite{Bhardwaj:2018jgp}, an F-theory compactification without $O7^+$-planes realizing the same massless spectrum. As a consequence, we conclude that for such models, it is a priori impossible to distinguish between the frozen and unfrozen phases purely from the non-abelian gauge algebras and massless matter content in the low energy effective physics. 

In this section, we will use our concrete F-theory models to propose an answer to Question~\ref{q:distinguish} for $6$d F-theory compactifications. We will exploit the distinction in the moduli spaces of two F-theory compactifications with the same massless spectrum, with one in the unfrozen phase, and with the other in the frozen phase. Our main claim is the following:
\begin{claim}\label{clm:limit}
Let $\mathcal{T}_1, \mathcal{T}_2$ be two $6d$ $\, \mathcal{N} = (1,0)$ supergravity theories with the same massless spectrum, engineered via a frozen/unfrozen F-theory compactification respectively. Assume that there exists an $\mathfrak{sp}(n)$ factor in both theories, which is localized on a frozen $7$-brane in the case of $\mathcal{T}_1$. 

Then $\mathcal{T}_2$ admits a limit in moduli space enhancing $\mathfrak{sp}(n)$ to $\mathfrak{so}(4n)$, while the analogous limit for $\mathcal{T}_1$ realizes a $(4,6)$-divisor.
\end{claim}

For simplicity, we will illustrate the claim in the example of Section~\ref{sec:aspinwallgross}, but we note that this can also be done for the examples computed in Sections~\ref{sec:1frozen} and \ref{sec:2frozen}. 

In Section~\ref{sec:aspinwallgross}, we considered the following Tate model on the base $B = \mathbb{F}_1$.
\[
a_1\sim p_1, \ a_2 \sim p_2, \ a_3  = 0, \ a_4 \sim e^8 p_4, \ a_6 = 0
\]
Letting $e,q$ denote the class of the exceptional and fiber classes respectively, we note that $p_1 \in \vert 2e + 2q \vert$, $p_2 \in \vert 4e + 4q \vert$, and $p_4 \in \vert 12 q \vert$. We consider the specialization of the Tate coefficient by the replacement
\[
a_2 \sim p_2 \mapsto a_2' \sim e p_2'
\]
inducing the following specialization of the corresponding Weierstrass model
\begin{align*}
f &= \frac{1}{3} (-p_2^2 + 3 e^8 p_4) \quad \mapsto \quad f' = \frac{1}{3} e^2 (-p_2^2 + 3 e^6 p_4)\\
g&= -\frac{1}{27} p_2(2p_2^2 - 9e^8 p_4)\quad \mapsto\quad g' = \frac{1}{27} e^3 p_2 (-2 p_2^2 + 9 e^6 p_4)  \\
\Delta &= -e^{16} p_4^2 (p_2^2 - 4e^8 p_4)\quad \mapsto\quad \Delta' = e^{18} p_4^2 (-p_2^2 + 4 e^6 p_4)
\end{align*}
In particular, we note that the $I_{16}^{ns}$ fiber of Section~\ref{sec:aspinwallgross} admits an enhancement to an $I_{12}^*$ fiber localized along the exceptional curve $e$ of $\mathbb{F}_1$, leading to an $\mathfrak{so}(32)$ algebra supported on $e$.

We now illustrate that a completely analogous enhancement in the frozen realization of the Aspinwall-Gross model leads to a $(4,6)$ divisor, and hence an infinite distance limit~\cite{Bershadsky:1996nh} in moduli space. We recall that the frozen Aspinwall-Gross model is given by the following Tate model on the base $B = \mathbb{F}_4$, as considered at the end of Section~\ref{sec:frozenconstruct}:
\[
a_1 = 0, \ a_2 \sim e p_2, \ a_3 = 0, \ a_4 \sim e^8 q^{48}, \ a_6 = 0
\]
which leads to the Weierstrass model:
\begin{align*}
f &= \frac{1}{3} e^2 (-p_2^2 + 3 e^6 q^{24})\\
g &= \frac{1}{27} e^3 p_2 (-2 p_2^2 + 9 e^6 q^{24})\\
\Delta &= e^{18} q^{48} (-p_2 + 2 e^3 q^{12}) (p_2 + 2 e^3 q^{12})
\end{align*}
We demand that $e$ supports an $\widehat{I}_{12}^*$-fiber, and hence supports a frozen $7$-brane. 
Subject to the prescription of Section~\ref{sec:frozenconstruct}, we define the gauge divisors
\[
\Sigma_1 = \frac{1}{2} e, \quad \Sigma_2 = 2q
\]
which leads to the massless spectrum of Section~\ref{sec:aspinwallgross} in the limit that all the $\mathfrak{sp}(1)$ singular fibers collide.

We now consider an analogous enhancement as performed in the unfrozen phase with the same massless spectrum; namely, we attempt to tune an additional orientifold onto the exceptional curve. This can be understood as attempting to tune the mass of open string states between the $7$-brane on the exceptional curve and the $7$-brane on the residual $I_1$ to zero. In the low energy effective physics, there is an identification of unlocalized neutral hypermultiplets between the two supergravity theories; the two tunings can then be identified with giving non-zero vacuum expectation values to the corresponding massless multiplets.

Geometrically, this corresponds to the following replacement:
\[
a_2 \sim ep_2 \mapsto a_2' \sim e^2 p_2'
\]
which leads to the Weierstrass model:
\begin{align*}
f&= \frac{1}{3} e^4 (-p_2^2 + 3 e^4 q^{24}) \\
g&= \frac{1}{27} e^6 p_2 (-2 p_2^2 + 9 e^4 q^{24}) \\
\Delta &= e^{20} q^{48} (-p_2 + 2 e^2 q^{12}) (p_2 + 2 e^2 q^{12})
\end{align*}
Clearly, this induces a $(4,6)$ divisor on the exceptional curve $e$, signaling an infinite distance limit. Such limits were first analyzed in~\cite{Cvetic:2022uuu}, which demonstrated how infinite distance limits can be realized in $8$d $\mathcal{N} = 1$ vacua through affinizations of the gauge algebra using string junctions. The appearance of $(4,6)$-points in $8$d F-theory lead to infinite dimensional Kac-Moody algebras, which result in decompactifications to either $9$ or $10$ dimensions. The analogous statements in $8$d F-theory with frozen singularities can then be obtained by embedding the $O7^+$-plane into the $9$d affine algebras. By studying our model in $6$-dimensions adiabatically, we expect that our infinite distance limit corresponds to the $9$d CHL vacua compactified on a $\mathbb{P}^1$ with an appropriate twist.
\section{Frozen LSTs, SCFTs and the Swampland}\label{sec:frozenscft}
In this section, we investigate whether $6d$ $\mathcal{N} = (1,0)$ supergravity theories which have not been constructed in ordinary F-theory, can be realized in the frozen phase of F-theory. Our main source of examples stems from a compilation~\cite{Bhardwaj:2019hhd} of local embeddings of new little string theories and superconformal field theories using frozen $7$-branes. Our goal is to study whether such local models can be completed into a compact F-theory model with frozen $7$-branes.

For unfrozen $6$d superconformal field theories, only a number of top-down constructions and constraints based on global symmetries~\cite{DelZotto:2014fia,Braun:2021sex} are known for the existence of compact embeddings. Nevertheless, we emphasize that for the new SCFT's realized in the frozen phase, most of the models do not admit a compact embedding based purely on their intersection constraints, independent of the ranks of the individual gauge algebras.

The gauge theory sectors summarized in~\cite[Section 3.1]{Bhardwaj:2019hhd} which do not admit a conventional F-theory embedding are given by the following:
\begin{enumerate}
\item
$\mathfrak{su}(n)$ with $1$ symmetric hypermultiplet and $n-8$ fundamental hypermultiplets
\item
$\mathfrak{su}(n)$ with $2n$ fundamental hypermultiplets, with a subset transforming under $\mathfrak{so}(m)$
\item
$\mathfrak{sp}(n)$ with $2n + 8$ fundamental hypermultiplets, with $3$ subsets transforming under three factors $\mathfrak{so}(m_1), \mathfrak{so}(m_2), \mathfrak{so}(m_3)$
\end{enumerate}
Our central claim can be summarized as follows:
\begin{claim}\label{clm:cases}
Case~$(1)$ never admits a compact embedding, but cases~$(2), (3)$ do admit compact embeddings in frozen F-theory models
\end{claim}
In Section~\ref{sec:6dtheories}, we review our conventions for specifying gauge theory chains, and we summarize a critical Lemma. In Section~\ref{sec:swampland}, we use the Kodaira condition to demonstrate that most of the above local embeddings can not be completed into a global model. In Section~\ref{sec:constructnew}, we demonstrate the existence of new $6d$ $\mathcal{N} = (1,0)$ supergravity theories via explicit constructions with frozen $7$-branes in F-theory. In the process, we explicitly realize an argument~\cite{Tarazi:2021duw} exploiting unitarity of string worldsheet theories in a compact F-theory model, which may be of independent mathematical interest.
\subsection{6d $\mathcal{N} = (1,0)$ theories}\label{sec:6dtheories}
We will primarily be interested in whether certain $6d$ $\mathcal{N} = 1$ gauge theory sectors can be completed into supergravity models. In this section, we review our conventions and our central argument.

%We note that results are complementary to~\cite{Tarazi:2021duw}. {\color{red}Comment on choice of lattice in string construction}
Let $B$ be a smooth complex algebraic surface. We will specify a local gauge theory sector through its non-abelian gauge algebra and matter content. This will be portrayed through intersections between singular Kodaira fibers yielding non-abelian gauge algebras with diagrams of the following type
%\[
%\begin{tikzcd}[row sep = -4pt]
%K_1(n_1)\arrow[dash]{r}& K_2(n_2) \\
%C & D
%\end{tikzcd}
%\]
\[
\xymatrix@R = .1em{
 K_1(n_1)\ar@{-} [r]&  K_2\, (n_2)\\
    C & D
    }
\]
where $K_i$ denote the associated Kodaira type, $n_i$ denotes the negative of the self-intersection of the curve on which the fiber is supported, a line between two Kodaira fibers indicates that corresponding curves intersect transversely in $B$, and $C$ and $D$ denote the curves of support. If such a gauge theory sector indeed admits an F-theory embedding, then it must satisfy the following necessary condition.
\begin{lemma}\label{lem:kodaira}
Let $\Delta \subset B$ be the discriminant locus of an elliptic fibration, admitting a factorization into irreducible components
\[
\Delta = z_1^{m_1} \cdot \ldots \cdot z_k^{m_k} \widetilde{\Delta}
\]
with associated divisors $D_1, \ldots D_k$. Then every divisor $D_i$ satisfies the condition
\[
(-12K_B - m_i D_i) \cdot D_i \geq 0
\]
\end{lemma}

\subsection{F-theory swampland}\label{sec:swampland}
In this section, we prove that most of the new superconformal field theories (SCFTs) and little string theories (LSTs) engineered via frozen $7$-branes in~\cite{Bhardwaj:2019hhd} do not admit a compact embedding. Our main technique will be to analyze the intersecting Kodaira fibers from  \cite{Bhardwaj:2019hhd}, focusing on particular sub-blocks of such constructions. We begin by assuming that there exists an embedding into a compact, frozen F-theory model subject to the rules of \cite{Bhardwaj:2018jgp}. In particular, when each $\hat{I}_{4+n}^*$-fiber is flipped to an $I_{4+n}^*$-fiber, the resulting configuration should still be a consistent $6$d F-theory model. Our main technique will be to analyze the contradictions in the resulting model by applying Lemma~\ref{lem:kodaira}.

In the following three sections, we will analyze three particular sub-blocks which feature in every new frozen $7$-brane construction in~\cite{Bhardwaj:2019hhd}. We will find that most of these sub-blocks do not admit a compact embedding.
\subsubsection{Tangential Intersections}\label{subsec:tangent}
We begin by demonstrating that a gauge theory sector which engineers an $\mathfrak{su}(n)$ gauge theory with $1$ symmetric tensor and $n - 8$ fundamentals does not admit a compact embedding. We remark that this was demonstrated earlier in~\cite{Tarazi:2021duw}, though such arguments are not a priori, obviously independent of the choice of basis for the anomaly lattice. On the other hand, our argument is geometric and assumes an F-theory realization, and so we view our results as complementary.

We engineer this gauge theory using the following configuration of singular fibers
\begin{comment}
\[
\begin{tikzcd}[row sep = -4pt]
\widehat{I_4^*}\, (4) \arrow[equal,"t"]{r}& I_n^s\, (1) \\
C & D
\end{tikzcd}
\]
\end{comment}
\[
\xymatrix@R = .1em{
    \widehat{I_4^*}\,(4)\ar^t@{=} [r]&  I_n^s \,(1)\\
    C & D
    }
\]
where the double line denotes a tangential intersection. Subject to the rules of Section~\ref{sec:frozenconstruct}, there is a single gauge divisor $\Sigma = C$ supporting an $\mathfrak{su}(n)$ gauge algebra. In particular, this configuration leads to an $\mathfrak{su}(n)$ gauge theory with $1$ symmetric and $n-8$ fundamentals.

We claim that the above configuration cannot be realized on any compact algebraic surface $B$, but for a rather subtle reason. Indeed, let $C, D \subset B$ be $(-4)$ and $(-1)$ curves respectively, intersecting tangentially at a single point. Consider a contraction $\pi \colon B \rightarrow B_0$, blowing down the $(-1)$-curve $D$. Then $\pi(C) \subset B_0$ is a cuspidal curve with self-intersection $0$. In particular, its corresponding linear system induces a morphism $| \pi(C) | \colon B_0 \rightarrow \mathbb{P}^1$, with generic fiber a smooth elliptic curve. On the other hand, any elliptic ruled surface giving a $6$d $(1,0)$ F-theory model must be a rational elliptic surface, and we claim that this cannot be the case.

To see the claim, we note that the anti-canonical bundle of $B_0$ satisfies $-K_{B_{0}}  = q$, with $q$ the fiber class. Let $f_0$ denote the cuspidal curve $\pi(C)$. Then the induced Weierstrass model on $B_0$ must satisfy the following
\[
f = f_0^2 \widetilde{f}, \ g = f_0^3 \widetilde{g}, \ \Delta = f_0^{10} \widetilde{\Delta} 
\]
where the vanishing loci of $\widetilde{f}, \widetilde{g}, \widetilde{\Delta}$ are localized on distinct fibers of $B_0$ not intersecting $f_0$. Blowing up the singular point on $f_0$ with exceptional divisor $e$, the proper transform is given by the class $\widetilde{f}_0 = \pi^* f_0 - 2 e$. We claim that the singular fiber localized on the exceptional divisor $e$ must be non-split.

Indeed, the splitness condition is given by checking the quantity
\[
\frac{9g}{2f} = \frac{\widetilde{f}_0^3 \widetilde{g}}{\widetilde{f}_0^2 \widetilde{f}}\bigg|_{e = 0} = \widetilde{f}_0 \frac{\widetilde{g}}{\widetilde{f}}\bigg|_{e= 0}
\]
where $\widetilde{f}_0$ denotes the proper transform of the class on the blowup. On the other hand, as the vanishing loci of $\widetilde{g}, \widetilde{f}$ correspond to distinct fibers on $B_0$, their proper transform on $B$ evaluated along $e = 0$ cannot contain any powers of the proper transform $f_0$. Thus we conclude that this quantity cannot factorize, and the singular fiber localized on the $(-1)$-curve must be non-split.

We further note that a compact engineering of such a model is possible in the case $n = 8$ realizing the following intersection.
\begin{comment}
\begin{equation*}
\begin{tikzcd}
\hat{I_4^*}\, (4) \arrow[equal,"t"]{r}& I_8^{ns}\, (1)
\end{tikzcd}
\end{equation*}
\end{comment}
\[
\xymatrix@R = .1em{
    \widehat{I_4^*}\,(4)\ar^t@{=} [r]&  I_8^s \,(1)
    }
\]
\subsubsection{$SU(n)$ on $(-2)$-curves}\label{subsec:triple}
We next consider case~$(2)$ of Claim~\ref{clm:cases}. A simple example of a configuration of singular fibers giving rise to such a gauge theory is given by the following
\begin{comment}
\[
\begin{tikzcd}
\widehat{I_4^*}(4)\arrow[dash]{r}& I_{m}^{ns}(1)(C) \arrow[dash]{r} \arrow[dash]{d}& I_n^{s}(2)(D) \\
&\hat{I}_{4}^*(4)&
\end{tikzcd}
\]
\end{comment}
\[
\xymatrix{
\widehat{I_4^*}(4)\ar@{-}[r]& I_{m}^{ns}(1)(C) \ar@{-}[r] \ar@{-}[d]& I_n^{s}(2)(D) \\
&\hat{I}_{4}^*(4)&
    }
\]
Using the rules of Section~\ref{sec:frozenconstruct}, we consider the gauge divisors $\Sigma_1 = 2C$, $\Sigma_2 = D$, supporting $\mathfrak{so}(m)$ and $\mathfrak{su}(n)$ respectively. Altogether, this leads to the following massless multiplets:
\begin{table}[h!]
\centering
\begin{tabular}{c | c}
\toprule
gauge algebra &  $\mathfrak{so}(m) \oplus \mathfrak{su}(n) $\\
\midrule
charged hypermultiplets & $ \mathbf{m} \otimes \mathbf{n} $\\ 
& $(2n - m)\times \mathbf{n}$\\ & $(m - 8 -n) \times \mathbf{m}$ \\

\bottomrule
\end{tabular}
\caption{Massless content of a $6d$ $\mathcal{N} = (1,0)$ gauge theory which cannot be engineered in the unfrozen phase of F-theory}
\label{tab:}
\end{table}

We briefly review the argument for why such a gauge theory cannot be realized in the unfrozen phase of F-theory. In a conventional F-theory compactification, we must have the $\mathfrak{so}(m)$ supported on a $(-4)$ curve $C_1$ and the $\mathfrak{su}(n)$ supported on a $(-2)$-curve $C_2$, such that $C_1$ and $C_2$ intersect transversely. In particular, the Weierstrass polynomial $f$ must take the following form
\[
f = c_1^2 \widetilde{f}, \quad \widetilde{f} \in -4K - 2C_1
\]
where $c_1 = 0$ defines the curve $C_1$. On the other hand, this implies that $\widetilde{f} \cdot C_2 = (-4K - 2C_1) \cdot C_2 = -2$, and thus $f$ must vanish to order at least $1$ along $C_2$. Together with a similar argument for the Weierstrass polynomial $g$, we conclude that this contradicts the assumption of $C_2$ supporting an $\mathfrak{su}(n)$ algebra.

Finally, we argue that many SCFTs and LSTs constructed in the frozen phase do not admit a compact embedding. We consider a more generalized sub-block given by the following
 %In lieu of a geometric analysis, we will use gravitational anomaly cancellation together with consistency in the unfrozen F-theory model as the main input. We recall that this implies the following constraint:
%\[
%H - V + 29T = 273
%\]
\begin{comment}
\[
\begin{tikzcd}
I_k^*(4)(D_2)\arrow[dash]{r}& I_{m}^{ns}(1)(D_3) \arrow[dash]{r} \arrow[dash]{d}& I_n^{s}(2)(D_4)\\
&\widehat{I_4^*}(4)(D_1)&
\end{tikzcd}
\]
\end{comment}
\[
\xymatrix{
I_k^*(4)(D_2)\ar@{-}[r] & I_{m}^{ns}(1)(D_3)\ar@{-}[r] \ar@{-}[d]& I_n^{s}(2)(D_4)\\
&\widehat{I_4^*}(4)(D_1)&
    }
\]
The discriminant must then factorize into the following form:
\[
\Delta = z_1^{10} z_2^{6+k}z_3^{m} z_4^n\widetilde{\Delta}
\]
The intersection number of the residual discriminant with the divisor $D_1$ is given by:
\begin{align*}
\widetilde{\Delta} \cdot D_1 &= (-12K - 10D_1 - (6+k) D_2 -m D_3 - nD_4) \cdot D_1 \\
&= -24 + 40 + 0 - m + 0 = 16 - m
\end{align*}
Similarly, we obtain
\[
\widetilde{\Delta} \cdot D_2 = 4k - m, \quad \widetilde{\Delta} \cdot D_3 = -4 - k + m - n, \quad \widetilde{\Delta} \cdot D_4 = -m + 2n
\]
Applying Lemma~\ref{lem:kodaira}, we demand that $\widetilde{\Delta} \cdot D_i \geq 0$ for all $i$, and this implies that the only valid solution is $m = 16, n = 8, k = 4$, in which case $\widetilde{\Delta} \cdot D_i = 0$ for all $i$.% {\color{red} double check this one last time}

In Section~\ref{sec:susoint}, we will explicitly construct a compact embedding of such a model.
%In the unfrozen setting, we must have the following gauge algebras
%\[
%\begin{tikzcd}
%\mathfrak{so}(16) \arrow[dash]{r} & \mathfrak{sp}(8) \arrow[dash]{r}\arrow[dash]{r} \arrow[dash]{d}& \mathfrak{su}(8) \\
%& \mathfrak{so}(16)&
%\end{tikzcd}
%\]
%where a line between $\mathfrak{so}$ and $\mathfrak{sp}$ factors denote a half-bifundamental, and the dash between $\mathfrak{sp}$ and $\mathfrak{su}$ factor denotes a single bifundamental. A straightforward calculation implies that there can only be bi-fundamental matter.

%On the other hand, two blowdowns of the $(-1)$ and $(-2)$ curve successively, yields an F-theory base with two $(-2)$ curves intersecting tangentially. This cannot happen on a blowup of $\mathbb{F}_n$ at three points, and we conclude {\color{red} add details here}. Moreover, we note that this gauge theory cannot be coupled to any other gauge theory sector by sharing bi-fundamental matter, and hence cannot be embedded into any compact F-theory model. {\color{red} add details here}
\subsubsection{$SU(n)$ and triple intersections}
In this section, we consider case~$(3)$ of Claim~\ref{clm:cases}. We first illustrate a simple configuration of singular fibers giving rise to such a gauge theory:
\begin{comment}
\[
\begin{tikzcd}
&&\widehat{I_{4}^*} (4) (D_7)\arrow[dash]{d} && \\ 
&&I_{c}^{ns}(1)(D_6) \arrow[dash]{d} && \\
\widehat{I_4^*} (4)(D_1) \arrow[dash]{r} & I_{a}^{ns} (1)(D_2) \arrow[dash]{r}& \widehat{I_k^*}(4)(D_3) \arrow[dash]{r}&I_{b}^{ns}(1)(D_4) \arrow[dash]{r} & \widehat{I_{4}^*}(4)(D_5)\\
\end{tikzcd}
\]
\end{comment}
\[
\xymatrix{
&&\widehat{I_{4}^*} (4) (D_7)\ar@{-}[d]&& \\ 
&&I_{c}^{ns}(1)(D_6)\ar@{-}[d] && \\
\widehat{I_4^*} (4)(D_1) \ar@{-}[r] & I_{a}^{ns} (1)(D_2) \ar@{-}[r]& \widehat{I_k^*}(4)(D_3) \ar@{-}[r]&I_{b}^{ns}(1)(D_4) \ar@{-}[r] & \widehat{I_{4}^*}(4)(D_5)\\
    }
\]
Again, using the rules of Section~\ref{sec:frozenconstruct}, we consider the assignment of gauge divisors 
\[
\Sigma_1 = 2D_2, \ \Sigma_2 = \frac{1}{2}D_3, \ \Sigma_3 = 2D_4, \ \Sigma_4 = 2D_6
\]
supporting the gauge algebras $\mathfrak{so}(a), \mathfrak{sp}(k-4), \mathfrak{so}(b), \mathfrak{so}(c)$ respectively. Altogether, this leads to the massless content in Table~\ref{tab:triple}.
\begin{table}[h!]
\centering
\begin{tabular}{c | c}
\toprule
gauge algebra &  $ \mathfrak{so}(a) \oplus \mathfrak{so}(b) \oplus \mathfrak{so}(c) \oplus \mathfrak{sp}(k-4)$\\
\midrule
charged hypermultiplets & $\frac{1}{2}\times (\mathbf{a}, \mathbf{2k-8}),\ \frac{1}{2}\times(\mathbf{b}, \mathbf{2k-8}),\ \frac{1}{2}\times(\mathbf{c}, \mathbf{2k-8}) $\\
&$ ((a - 8) - (k-4) )\times \mathbf{a}, \ ((b - 8) - (k-4))\times \mathbf{b}, \ ( (c - 8) - (k-4))\times \mathbf{c}$ \\
&$  (2(k-4) + 8 - \frac{a+b+c}{2} )\times \mathbf{2k-8} $ \\
\bottomrule
\end{tabular}
\caption{Massless content of a $6d$ $\mathcal{N} = (1,0)$ gauge theory which cannot be engineered in the unfrozen phase of F-theory}
\label{tab:triple}
\end{table}
%\begin{itemize}
%\item
%$\mathfrak{so}(a) \times \mathfrak{so}(b) \times \mathfrak{so}(c) \times \mathfrak{sp}(k-4)$
%\item
%half-bifundamentals $\frac{1}{2}(\mathbf{a}, \mathbf{2k-8}),\frac{1}{2}(\mathbf{b}, \mathbf{2k-8}),\frac{1}{2}(\mathbf{c}, \mathbf{2k-8})$, $(a - 8) - (k-4)$ fundamentals $\mathbf{a}$,$(b - 8) - (k-4)$ fundamentals $\mathbf{b}$, $(c - 8) - (k-4)$ fundamentals $\mathbf{c}$, and $2(k-4) + 8 - \frac{a+b+c}{2}$ fundamenetals $\mathbf{2k-8}$
%\end{itemize}
The non-existence of a conventional F-theory realization follows by a similar argument as in the previous section, which we briefly illustrate. To realize such massless content in the unfrozen phase of F-theory, we must have $\mathfrak{so}(a),\mathfrak{so}(b),\mathfrak{so}(c)$ supported on individual $(-4)$-curves $C_1, C_2, C_3$ transversely intersecting the curve $C_4$ which supports $\mathfrak{sp}(k-4)$. In particular, the Weierstrass polynomial $f$ must take the following form
\[
f = c_1^2 c_2^2 c_3^2 \widetilde{f}, \quad \widetilde{f} \in -4K - 2C_1 - 2C_2 - 2C_3
\]
In particular, we have the intersection $\widetilde{f} \cdot C_4 = (-4K - 2C_1 - 2C_2 - 2C_3) \cdot C_4 = -2$, and hence $\widetilde{f}$ must vanish to at least order $2$ along $C_4$. Arguing similarly with the Weierstrass polynomial $g$, this yields a contradiction with the assumption that $C_4$ supports an $I_{2k-8}^{ns}$-fiber.

Finally, we consider the general sub-block given by the following intersections:
\begin{comment}
\[
\begin{tikzcd}
&&I_{m}^* (4) (D_7)\arrow[dash]{d} && \\ 
&&I_{c}^{ns}(1)(D_6) \arrow[dash]{d} && \\
\hat{I}_4^* (4)(D_1) \arrow[dash]{r} & I_{a}^{ns} (1)(D_2) \arrow[dash]{r}& I_{k}^*(4)(D_3) \arrow[dash]{r}&I_{b}^{ns}(1)(D_4) \arrow[dash]{r} & I_{l}^*(4)(D_5)\\
\end{tikzcd}
\]
\end{comment}
\[
\xymatrix{
&&I_{m}^* (4) (D_7)\ar@{-}[d] && \\ 
&&I_{c}^{ns}(1)(D_6) \ar@{-}[d]&& \\
\hat{I}_4^* (4)(D_1) \ar@{-}[r]& I_{a}^{ns} (1)(D_2) \ar@{-}[r]& I_{k}^*(4)(D_3) \ar@{-}[r]&I_{b}^{ns}(1)(D_4) \ar@{-}[r]& I_{l}^*(4)(D_5)\\
    }
\]
We analyze the discriminant locus, which takes the following form
\[
\Delta = z_1^{10} z_2^{a} z_{3}^{6+k} z_4^{b} z_5^{6+l} z_6^{c} z_7^{6+m} \widetilde{\Delta}
\]
Computing the intersection of $\widetilde{\Delta}$ with $D_1$, we obtain the following:
\begin{align*}
\widetilde{\Delta} \cdot D_1 &= (-12K - 10 D_1 - a D_2 - (6+k)D_3 - b D_4 - (6+l)D_5 - c D_6 - (6+m) D_7 ) \cdot D_1 \\
&= 16 - a
\end{align*}
Similarly, we obtain the following intersections
\begin{align*}
\widetilde{\Delta} \cdot D_2 &= -4 + a - k, \quad \widetilde{\Delta} \cdot D_3  = -a - b - c + 4 k, \quad \widetilde{\Delta} \cdot D_4 =  b - k - l\\
\widetilde{\Delta} \cdot D_5 &= -b + 4l, \quad \widetilde{\Delta} \cdot D_6 =c - k - m, \quad \widetilde{\Delta} \cdot D_7 =-c + 4 m
\end{align*}
Applying Lemma~\ref{lem:kodaira} again and imposing the non-negativity of all the above inequalities implies the following unique solution $a = b = c = 16,\  l = m = 4, \ k = 12$. 

In Section~\ref{sec:3soint}, we will explicitly construct a compact embedding of such a model.
%On the other hand, to satisfy gauge anomaly cancellation, it must be the case that $\mathfrak{so}(16)$ carries precisely $8$ fundamentals. On the other hand, anomaly cancellation also requires that there is a half bi-fundamental at the intersection with each $\mathfrak{sp}$ factor. This is clearly in contradiction at the middle node.
\subsection{New massless spectra in $6$d F-theory}\label{sec:constructnew}
In the previous section, we focused on ruling out infinite families of potential compact embeddings of frozen $7$-brane configurations by applying Lemma~\ref{lem:kodaira} to universal sub-blocks of such constructions. Nevertheless, we discovered special cases which were consistent with such a condition.

In this section, we will demonstrate that such cases do indeed admit compact embeddings with frozen $7$-branes, leading to $6d$ $\mathcal{N} = (1,0)$ vacua with new massless gauge and matter content. Along the way, we find a potential infinite family in Section~\ref{sec:infinitefam}, which we rule out by relying on string unitarity arguments introduced in~\cite{Kim:2019vuc,Tarazi:2021duw}. As a consequence, we find a concrete geometric realization of their argument, which may be of independent mathematical interest. 

\subsubsection{Intersections of $\mathfrak{su}(n)$ with $\mathfrak{so}(m)$}\label{sec:susoint}
We first consider the special case realized in Section~\ref{subsec:triple}, given by the following configuration
\begin{equation}\label{frozendiagram}
%\begin{tikzcd}
%\widehat{I}_4^*(4)\arrow[dash]{r}& I_{16}^{ns}(1) \arrow[dash]{r} \arrow[dash]{d}& I_8^{s}(2) \\
%&\widehat{I}_{4}^*(4)&
%\end{tikzcd}
\xymatrix{
\widehat{I}_4^*(4)\ar@{-}[r]& I_{16}^{ns}(1) \ar@{-}[r] \ar@{-}[d]& I_8^{s}(2) \\
&\widehat{I}_{4}^*(4)&
    }
\end{equation}
which realizes the following chain of gauge algebras
\[
%\begin{tikzcd}
%\mathfrak{so}(8)(4) \arrow[dash]{r} & \mathfrak{su}(8)(2)
%\end{tikzcd}
\xymatrix{
\mathfrak{so}(8)(4)\ar@{-}[r]& \mathfrak{su}(8)(2)
    }
\]

We will construct a compact realization of such a model using the Tate form described in Section~\ref{sec:tate}. Let $B$ be the Hirzebruch surface $\mathbb{F}_2$ with a unique $(-2)$-curve $E$ defined by $e = 0$. Recall that the anti-canonical class satisfies $-K_B = 2e + 4f$ and consider the following tuning
\[
a_1 \sim q e, \ a_2 \sim q e p_2, \ a_3 = 0, \ a_4 \sim e^4 q^4, \  a_6 = 0
\]
where $q = 0$ defines a general curve $Q$ in the linear system $\vert e + 4f \vert$. In particular, we note that the polynomial $p_2$ must take values in the system $\vert 2e + 4f\vert$. This leads to the following Weierstrass model:
\begin{align*}
f &= \frac{1}{48} e^2 q^2 (-16 p_2^2 - 8 e p_2 q + 47 e^2 q^2) \\
g&= \frac{1}{864} e^3 q^3 (4 p_2 + e q) (-16 p_2^2 - 8 e p_2 q + 71 e^2 q^2) \\
\Delta &= \frac{1}{16} e^{10} q^{10} (-4 p_2 + 7 e q) (4 p_2 + 9 e q)
\end{align*}
We may assume that $Q$ defines a curve in $\mathbb{F}_2$ tangentially intersecting the zero section $E$. In addition, we note that generically, the curve defined by $p_2  = 0$ intersects with $Q$ at $8$ points, as we have $(2e + 4f) \cdot (e+4f) = 8$. From the Weierstrass model, we conclude that there exists $8$ $(4,6)$ points localized along the intersections $Q \cap p_2$, as well as a single $(4,6)$-point localized at the tangential point of intersection $Q \cap E$.

We perform the following phase transitions, moving along the tensor branch of an F-theory compactification specified by the above Weierstrass model. We first blowup all $8$ $(4,6)$ points corresponding to the intersection of $p_2 = 0$ with $Q$. As $Q$ has self-intersection $6$ on $\mathbb{F}_2$, its proper transform has self-intersection $0$. Next, we perform a blowup at the intersection $Q \cap E$, leading to an exceptional curve $E_1$ supporting an $I_8^{ns}$-fiber, as can be seen from the form of the discriminant $\Delta$. As this is a tangential intersection, the proper transforms of $Q$ and $E$ intersect $E_1$ at a common point. Blowing up this common point once more leads to an exceptional curve supporting an $I_{16}^{ns}$-fiber. A direct computation of the self-intersections leads precisely to diagram~(\ref{frozendiagram}), and we note that the proper transform of $E_1$ must support a split fiber since it has self-intersection $(-2)$. The resulting discriminant takes the form
\[
\Delta = z_1^{10} z_2^{10}z_3^{16} z_4^{8}\widetilde{\Delta}
\]
We note that $\tilde{\Delta} \cdot D_i = 0$ for all $i$, and hence there are no further $(4,6)$-points in the resulting model.
\subsubsection{Intersections of $\mathfrak{sp}$ with $\mathfrak{so}$}\label{sec:3soint}
Next, we construct an explicit Tate model for the following:
\begin{equation}\label{frozendiagramthree}
\xymatrix{
&& \widehat{I}_4^*(4)\ar@{-}[d]& &\\
&&I_{16}^{ns}(1)\ar@{-}[d]& &\\
\widehat{I}_4^*(4) \ar@{-}[r]&I_{16}^{ns}(1)\ar@{-}[r]& \widehat{I}_{12}^{*}(4)\ar@{-}[r]& I_{16}^{ns}(1) \ar@{-}[r]& \widehat{I}_4^*(4)
    }
\begin{comment}
\begin{tikzcd}
&& \widehat{I}_4^*(4)\arrow[dash]{d}& &\\
&&I_{16}^{ns}(1)& &\\
\widehat{I}_4^*(4) \arrow[dash]{r}&I_{16}^{ns}(1)\arrow[dash]{r}& \widehat{I}_{12}^{*}(4) \arrow[dash]{r} \arrow[dash]{u}& I_{16}^{ns}(1) \arrow[dash]{r} & \widehat{I}_4^*(4)
\end{tikzcd}
\end{comment}
\end{equation}
Let $B$ be the projective plane $\mathbb{P}^2$. Consider the following tuning of the Tate coefficients
\[
a_1 \sim h_1 h_2 h_3, \ a_2 \sim h_1 h_2 h_3 p_2, \ a_3 = 0, a_4  \sim h_1^4 h_2^4 h_3^4, \ a_6 = 0
\]
where we take $h_1, h_2, h_3$ to be three lines intersecting at a common point and we note that $p_2$ is a generic polynomial of degree $3$. The corresponding Weierstrass model takes the following form
\begin{align*}
f&= \frac{1}{48} h_1^2 h_2^2 h_3^2 (47 h_1^2 h_2^2 h_3^2 - 8 h_1 h_2 h_3 p_2 - 16 p_2^2) \\
g &= \frac{1}{864} h_1^3 h_2^3 h_3^3 (h_1 h_2 h_3 + 4 p_2) (71 h_1^2 h_2^2 h_3^2 - 
   8 h_1 h_2 h_3 p_2 - 16 p_2^2) \\
\Delta &= \frac{1}{16} h_1^{10} h_2^{10} h_3^{10} (7 h_1 h_2 h_3 - 4 p_2) (9 h_1 h_2 h_3 + 4 p_2)
\end{align*}
We first note that the cubic $Q$ defined by $p_2 =0$ intersects each line $h_i$ at $3$ $(4,6)$ points. Blowing up all the intersections $Q \cap h_i$, the proper transforms of the lines $h_i$ are $(-2)$-curves. 

We next perform the following phase transitions, moving further along the tensor branch of the corresponding F-theory model. We first blowup the triple intersection $h_1 \cap h_2 \cap h_3$, whose exceptional divisor $E$ supports an $I_{12}^*$ fiber as can be seen from the above Weierstrass model. We note that the intersections $E \cap h_i$ are also $(4,6)$-points, and blowing up these three points leads to three additional exceptional divisors $E_i$, each supporting an $I_{16}^{ns}$-fiber. We conclude by noting that this produces precisely Figure~(\ref{frozendiagramthree}), and that the residual discriminant does not intersect with any divisors in the diagram.
\subsubsection{An infinite family}\label{sec:infinitefam}
Finally, we note that there exists an additional, a priori infinite family of such models with the following intersecting Kodaira fibers
\begin{equation}\label{frozendiagramthree}
\begin{comment}
\begin{tikzcd}
&& \widehat{I}_{m+4}^*(4)\arrow[dash]{d}& &\\
&&I_{4m  +16}^{ns}(1)& &\\
\widehat{I}_{m+4}^*(4) \arrow[dash]{r}&I_{4m+16}^{ns}(1)\arrow[dash]{r}& \widehat{I}_{3m + 12}^{*}(4) \arrow[dash]{r} \arrow[dash]{u}& I_{4m + 16}^{ns}(1) \arrow[dash]{r} & \widehat{I}_{m+4}^*(4)
\end{tikzcd}
\end{comment}
\xymatrix{
&& \widehat{I}_{m+4}^*(4)\ar@{-}[d]& &\\
&&I_{4m  +16}^{ns}(1)\ar@{-}[d]& &\\
\widehat{I}_{m+4}^*(4)\ar@{-}[r]&I_{4m+16}^{ns}(1)\ar@{-}[r]& \widehat{I}_{3m + 12}^{*}(4)\ar@{-}[r] & I_{4m + 16}^{ns}(1)\ar@{-}[r]& \widehat{I}_{m+4}^*(4)
    }
\end{equation}
We briefly provide an alternative argument to show that such models can only exist for $m = 0$, demonstrating that our construction in Section~\ref{sec:3soint} is sharp. Based on the analysis of~\cite{Kim:2019vuc}, \cite{Tarazi:2021duw} used consistency of non-instantonic string probes in $6$d supergravity backgrounds to prove the finiteness of many classes of models. We refer to~\cite{Tarazi:2021duw} for details, and we briefly summarize the necessary aspects in applying the argument to our setting.

Our geometric construction in Section~\ref{sec:3soint} gives a natural embedding of the anomaly vectors of the gauge algebras associated with the unflipped version of~(\ref{frozendiagramthree}). A $D3$-brane wrapped on an effective curve $C$ in a compact algebraic surface $B$ yields a $2d$ $\mathcal{N} = (0,4)$ supersymmetric worldsheet theory in the noncompact $6$-dimensional spacetime. The condition that this corresponds to a non-instantonic string with non-negative tension translates into the following geometric constraints:
\begin{align*}
C^2 \geq -1, \quad k_l = C\cdot (C+ K_B) + 2 \geq 0, \quad k_i = C \cdot D_i \geq 0
\end{align*}
where $D_i$ denotes the divisors supporting the non-abelian gauge algebras in~(\ref{frozendiagramthree}). Unitarity of the string worldsheet theory then imposes the following constraint
\begin{equation}\label{eq:unitarity}
\sum\limits_{i} c_{G_i} \leq c_L = 3 C^2 - 9 C\cdot K + 2, \quad c_{G_i} \coloneqq k_i \frac{dim G_i}{k_i + h_i^\vee}
\end{equation}
where $c_{G_i}$ denotes the central charge associated with the current algebra of the gauge algebra $G_i$, and $h_i^\vee$ denotes the dual Coxeter number of $G_i$.

Taking the three outer $(-4)$-curves in~(\ref{frozendiagramthree}), we recall that these were obtained as the proper transform of hyperplane classes on $\mathbb{P}^2$ by blowing up three points each on three intersecting lines. Choosing two out of these nine points such that these do not lie on the same of the three lines, we take a separate line passing through these two points on $\mathbb{P}^2$. We then take $C$ to denote the proper transform of this curve. A straightforward computation implies that $C$ only intersects with one of the outermost vertices corresponding to the Kodaira fiber $I_{m+4}^*$. Evaluating equation~(\ref{eq:unitarity}) with $C$, we obtain the bound
\[
c_{G} = 8 + m \leq 8
\]
Thus, it must be the case that $m= 0$ and this implies that our geometric construction of the model in Section~\ref{sec:3soint} is sharp. It would be interesting to interpret more generally, the condition on the unitarity of the string worldsheet as a geometric constraint on engineering singular limits of compact Weierstrass elliptic Calabi-Yau threefolds.
\subsubsection{Remaining models}
In this section, we list the remaining diagrams in~\cite{Bhardwaj:2019hhd} which admit a compact embedding in the frozen phase of F-theory. We note in passing that this can be done using the same techniques carried out in the previous sections, but we do not carry out the analysis here. 

In each entry, we list the configuration of Kodaira singular fibers and the corresponding equation in~\cite{Bhardwaj:2019hhd}. We refer the reader to the relevant section in~\cite{Bhardwaj:2019hhd} for the corresponding gauge theory configuration.
%We then give a brief argument for their construction, following the arguments of Section~\ref{sec:constructnew}.
\begin{enumerate}
\item
Equation~$(3.32)$
\begin{equation}\label{frozendiagram}
%\begin{tikzcd}
%I_6^s(2)\arrow[dash]{r}& I_{12}^{ns}(1) \arrow[dash]{r} \arrow[dash]{d}& I_2^{*s}(2) \\
%&\widehat{I}_{4}^*(4)&
%\end{tikzcd}
\xymatrix{
I_6^s(2)\ar@{-}[r]& I_{12}^{ns}(1) \ar@{-}[r] \ar@{-}[d]& I_2^{*s}(2) \\
&\widehat{I}_{4}^*(4)&
    }
\end{equation}
\item
Equation~$(3.34)$
\begin{equation}\label{frozendiagram}
%\begin{tikzcd}
%I_7^s(2)\arrow[dash]{r}& I_{14}^{ns}(1) \arrow[dash]{r} \arrow[dash]{d}& I_3^{*ns}(2) \\
%&\widehat{I}_{4}^*(4)&
%\end{tikzcd}
\xymatrix{
I_7^s(2)\ar@{-}[r]& I_{14}^{ns}(1) \ar@{-}[r] \ar@{-}[d]& I_3^{*ns}(2) \\
&\widehat{I}_{4}^*(4)&
    }
\end{equation}
\item
Equation~$(3.36)/(3.38)$
\begin{equation}\label{frozendiagram}
%\begin{tikzcd}
%I_4^s(2)\arrow[dash]{r}& I_{8}^{ns}(1) \arrow[dash]{r} \arrow[dash]{d}& I_0^{*ss/ns}(2) \\
%&\widehat{I}_{4}^*(4)&
%\end{tikzcd}
\xymatrix{
I_4^s(2)\ar@{-}[r]& I_{8}^{ns}(1)\ar@{-}[r] \ar@{-}[d]& I_0^{*ss/ns}(2) \\
&\widehat{I}_{4}^*(4)&
    }
\end{equation}
\item
Equation~$(3.50)/(3.52)$

\begin{equation}\label{frozendiagram}
%\begin{tikzcd}
%&&\widehat{I}_{4}^*(4)& &\\
%&&I_{12}^{ns}(1)\arrow[dash]{u}& &\\
%\widehat{I}_4^*(4)\arrow[dash]{r}&I_{12}^{ns}(1)\arrow[dash]{r}& \widehat{I}_{8}^{*s}(4) \arrow[dash]{r} \arrow[dash]{u}& I_{8}^{ns}(1)\arrow[dash]{r}& I_0^{*ss/ns}(2)
%\end{tikzcd}
\xymatrix{
&&\widehat{I}_{4}^*(4)& &\\
&&I_{12}^{ns}(1)\ar@{-}[u]& &\\
\widehat{I}_4^*(4)\ar@{-}[r]&I_{12}^{ns}(1)\ar@{-}[r]& \widehat{I}_{8}^{*s}(4) \ar@{-}[r]\ar@{-}[u]& I_{8}^{ns}(1)\ar@{-}[r]& I_0^{*ss/ns}(2)
    }
\end{equation}
\item
Equation~$(3.54)$

\begin{equation}\label{frozendiagram}
%\begin{tikzcd}
%&&\widehat{I}_{4}^*(4)& &\\
%&&I_{13}^{ns}(1)\arrow[dash]{u}& &\\
%\widehat{I}_4^*(4)\arrow[dash]{r}&I_{14}^{ns}(1)\arrow[dash]{r}& \widehat{I}_{9}^{*ns}(4) \arrow[dash]{r} \arrow[dash]{u}& I_{9}^{ns}(1)\arrow[dash]{r}& I_0^{*ns}(2)
%\end{tikzcd}
\xymatrix{
&&\widehat{I}_{4}^*(4)& &\\
&&I_{13}^{ns}(1)\ar@{-}[u]& &\\
\widehat{I}_4^*(4)\ar@{-}[r]&I_{14}^{ns}(1)\ar@{-}[r]& \widehat{I}_{9}^{*ns}(4) \ar@{-}[r] \ar@{-}[u]& I_{9}^{ns}(1)\ar@{-}[r]& I_0^{*ns}(2)
    }
\end{equation}
\item
Equation~$(3.57)$

\begin{equation}\label{frozendiagram}
%\begin{tikzcd}
%&&I_{2}^{*ns}(3)& &\\
%&&I_{8}^{ns}(1)\arrow[dash]{u}& &\\
%I_2^{*ns}(3)\arrow[dash]{r}&I_{8}^{ns}(1)\arrow[dash]{r}& \widehat{I}_{6}^{*s}(4) \arrow[dash]{r} \arrow[dash]{u}& I_{8}^{ns}(1)\arrow[dash]{r}& I_2^{*ns}(3)
%\end{tikzcd}
\xymatrix{
&&I_{2}^{*ns}(3)& &\\
&&I_{8}^{ns}(1)\ar@{-}[u]& &\\
I_2^{*ns}(3)\ar@{-}[r]&I_{8}^{ns}(1)\ar@{-}[r]& \widehat{I}_{6}^{*s}(4) \ar@{-}[r]\ar@{-}[u]& I_{8}^{ns}(1)\ar@{-}[r]& I_2^{*ns}(3)
    }
\end{equation}
\item
Equation~$(3.58)$
\begin{equation}\label{frozendiagram}
%\begin{tikzcd}
%&&I_{2}^{*s}(3)& &\\
%&&I_{11}^{ns}(1)\arrow[dash]{u}& &\\
%I_2^{*s}(3)\arrow[dash]{r}&I_{11}^{ns}(1)\arrow[dash]{r}& \widehat{I}_{9}^{*ns}(4) \arrow[dash]{r} \arrow[dash]{u}& I_{11}^{ns}(1)\arrow[dash]{r}& I_2^{*s}(3)
%\end{tikzcd}
\xymatrix{
&&I_{2}^{*s}(3)& &\\
&&I_{11}^{ns}(1)\ar@{-}[u]& &\\
I_2^{*s}(3)\ar@{-}[r]&I_{11}^{ns}(1)\ar@{-}[r]& \widehat{I}_{9}^{*ns}(4)\ar@{-}[r] \ar@{-}[u]& I_{11}^{ns}(1)\ar@{-}[r]& I_2^{*s}(3)
    }
\end{equation}
\end{enumerate}
\section{Conjectures on neutral hypermultiplets}\label{sec:neutrals}
In Sections~\ref{sec:ordinaryf} and~\ref{sec:frozenscft}, we analyzed a number of compact F-theory models with frozen $7$-branes. Such concrete analyses are potentially helpful in deducing general properties of frozen $7$-branes in $6$-dimensional compactifications. In this section, we analyze the simplest compact embeddings of frozen $7$-branes and make a simple observation based on gravitational anomaly cancellation.

\begin{conj}\label{conj:neutrals}
Consider a $6$d F-theory compactification with a frozen $7$-brane localized along a divisor $D$ with residual discriminant $\widetilde{\Delta}$. Then there exists at least $\frac{1}{2}D \cdot \widetilde{\Delta}$ localized neutral hypermultiplets.
\end{conj}
\noindent
We note that none of the compact examples considered thus far exhibit non-trivial intersections between the residual $I_1$-locus and a frozen $O7^+$-plane.

For a $6d$ $\mathcal{N} = (1,0)$ supergravity, we recall that gravitational anomalies impose the following constraint
\[
H - V + 29T = 273
\]
%For convenience, we define the quantity
%\[
%n(T) = 273 - (H-V) -29 T
%\]
where $H$ is the total number of hypermultiplets, charged and uncharged, $V$ is the number of vector multiplets, and $T$ is the number of tensor multiplets.

We consider an otherwise generic Weierstrass model on a compact algebraic surface $B$ with an $I_{n+4}^{*s}$ fiber localized along a divisor $D$ of self-intersection $(-4)$. In the unfrozen phase of F-theory, this realizes an $\mathfrak{so}(2n+16)$ gauge algebra with $2n + 8$ fundamental hypermultiplets. In the frozen phase, this realizes a $\mathfrak{sp(n)}$ gauge algebra with $2n + 8 $ fundamental hypermultiplets. For $6$d supergravity theories with such matter content, we compute the difference $H_{ch}-V$ for every value of $n$ in the following table. We note that for $n \geq 9$, there does not exist any number of tensor multiplets such that the resulting model satisfies gravitational anomalies. In the last row, we subtract the second row from the first row. Using gravitational anomalies, and assuming that the number of tensor multiplets is the same, this computes the difference in neutral hypermultiplets between the two models.

\begin{table}[h!]
\centering
\begin{tabular}{c|ccccccccc}
\toprule
n &0 &1 & 2 & 3 & 4 & 5 & 6 & 7 & 8 \\
\midrule
$(H_{ch}-V)({\mathfrak{so}(2n+16)})$&8& 27& 50& 77& 108 &143 &182& 225& 272\\
$(H_{ch}-V)({\mathfrak{sp}(n)})$&0& 17& 38& 63& 92& 125& 162& 203& 248 \\ 
$n_{\mathfrak{sp}(n)}- n_{\mathfrak{so}(2n+16)} $& 8 & 10& 12& 14& 16& 18& 20& 22& 24 \\ 
\bottomrule
\end{tabular}
\caption{Table of the difference between the number of charged hypermultiplets and vector multiplets of a generic F-theory model with an $\mathfrak{so}(2n+16)$ and an $\mathfrak{sp}(n)$ gauge algebra respectively. The difference in the counting of neutral hypermultiplets is encoded in the last row.}
\end{table}
In general, the total number of neutral hypermultiplets can be split into localized and unlocalized hypermultiplets. We expect that unlocalized neutral hypermultiplets should coincide with complex structure deformations of the total space of the Weierstrass model that do not deform the singularities, and hence that this should yield the same answer between the unfrozen and the frozen phase of F-theory on identical Calabi-Yau threefolds. Moreover, the intersections with the residual discriminant are given by the following
\[
\widetilde{\Delta} \cdot D  = (-12K - (10+ n)D) \cdot D = -24 + 4(10+n) = 16 + 4n
\]
In the unfrozen phase, there are $2n +8$ fundamental hypermultiplets, and so there is $1$ fundamental hypermultiplet for every double intersection\footnote{We thank Jiahua Tian for discussions on this point} with the residual discriminant. We conclude that in the frozen phase, there must be $1$ localized neutral hypermultiplet for every double intersection of the frozen $7$-brane with the residual discriminant.

It would be interesting to confirm such a result directly via a weakly coupled type IIB limit in the frozen phase of F-theory. The existence of such localized neutral hypermultiplets could perhaps be explained via a similar mechanism as in the $I_1$-conifold model explored in~\cite{Braun:2014nva} and \cite[Section 4.1]{Arras:2016evy}.
\section{Acknowledgements}
B.S. is grateful to Yuji Tachikawa and Alessandro Tomasiello for useful correspondences, to Jonathan Heckman, Paul Oehlmann, Fabian Ruehle, and Jiahua Tian for discussions, and to Fabian Ruehle and Paul Oehlmann for collaboration on a related project. We thank Paul Oehlmann for a reading of a preliminary draft. DRM is partially supported by the National Science Foundation Grant PHY-2014226.
\begin{comment}
\section{Conclusion}
Future directions:
\begin{enumerate}
\item
sufficient conditions and systematic construction of consistent frozen vacua
\item
phase transitions in the frozen phase via heterotic/type I
\end{enumerate}
\end{comment}
\begin{appendices}
\end{appendices}
\newpage
\bibliographystyle{JHEP}
\bibliography{refs}

\end{document}